\documentclass[11pt]{article}

\usepackage[utf8]{inputenc} 
\usepackage[T1]{fontenc}    
\usepackage{hyperref}       
\usepackage{url}            
\usepackage{booktabs}       
\usepackage{amsfonts}       
\usepackage{nicefrac}       
\usepackage{microtype}      
\usepackage[linesnumbered, boxed]{algorithm2e}
\usepackage{amsmath}
\usepackage{matharticle}
\usepackage{bm}
\usepackage{fullpage}
\usepackage{caption}

\usepackage{times}

\usepackage{subfigure}
\usepackage{algpseudocode}

\usepackage{bbm}
\usepackage[numbers]{natbib}
\usepackage{graphicx}

\let\marginpar\oldmarginpar

\usepackage[textsize=tiny]{todonotes}
\newcommand{\dan}[1]{\todo[color=blue!20]{#1}}

\newcommand{\tm}{\tau_{\max}}

\newcommand{\tg}{\widetilde{g}}
\newcommand{\fadd}{\texttt{fetch\&add}}

\newcommand{\norm}[1]{\left\|#1\right\|}

\newcommand{\lit}[1]{$\mathsf{#1}$}

\renewcommand{\paragraph}[1]{\vspace{0.1cm} \noindent\textbf{#1}\hspace{0.15em}}

\newtheorem*{rep@theorem}{\rep@title} \newcommand{\newreptheorem}[2]{%
\newenvironment{rep#1}[1]{%
\def\rep@title{\bf #2 \ref{##1} }%
\begin{rep@theorem} }%
{\end{rep@theorem} } }
\makeatother
\newreptheorem{theorem}{Theorem}
\newreptheorem{lemma}{Lemma}
\newreptheorem{corollary}{Corollary}

\title{The Convergence of Stochastic Gradient Descent \\ in Asynchronous Shared Memory}
\author{ Dan Alistarh \\ IST Austria \and Christopher De Sa \\ Cornell \and Nikola Konstantinov \\ IST Austria}
\date{}

\begin{document}

\maketitle

\begin{abstract}
Stochastic Gradient Descent (SGD) is a fundamental algorithm in machine learning, representing the optimization backbone for training several classic models, from regression to neural networks. Given the recent practical focus on distributed machine learning, significant work has been dedicated to the convergence properties of this algorithm under the inconsistent and noisy updates arising from execution in a distributed environment. However, surprisingly, the convergence properties of this classic algorithm in the standard shared-memory model are still not well-understood. 

In this work, we address this gap, and provide new convergence bounds for lock-free concurrent stochastic gradient descent, executing in the classic asynchronous shared memory model, against a strong adaptive adversary. Our results give improved upper and lower bounds on the ``price of asynchrony'' when executing the fundamental SGD algorithm in a concurrent setting. 
They show that this classic optimization tool can converge faster and with a wider range of parameters than previously known under asynchronous iterations. 
At the same time, we exhibit a fundamental trade-off between the maximum delay in the system and the rate at which SGD can converge, which governs the set of parameters under which this algorithm can still work efficiently. 
\end{abstract}





\section{Introduction}

\paragraph{Background.} The tremendous recent progress in machine learning can be explained in part through the availability of extremely vast amounts of data, but also through improved computational support for executing machine learning tasks at scale.
Perhaps surprisingly, some of the main algorithms driving this progress are have been known in some form or another for a very long time. One such tool is the classic stochastic gradient descent (SGD) optimization algorithm~\cite{RM51}, introduced by Robbins and Munro in the 1950s,
variants of which are currently the tool of choice for optimization in a variety of settings, such as image classification and speech recognition via neural networks, but also in fundamental data processing tools such as regression. 

In a nutshell, SGD works as follows. Given a $d$-dimensional function 
$f : \R^d \rightarrow \R$ which we want to minimize, 
 we have access to stochastic gradients $\tg(x)$ of this function at the point $x$, that is, random vectors with the property that  $\E [\tg(x) ] = \nabla f(x)$. 
SGD then consists in the iteration
$$x_{t + 1} = x_t - \alpha_t \tg(x_t),$$ 
\noindent where $x_t$ is a $d$-dimensional vector which we will henceforth call the \emph{model}, encoding our current beliefs about the data, and $\alpha_t$ is the step size, which controls how ``aggressive''  updates should be. 

A classic instance of this method, which we  consider in this paper, is the following. 
We are given a large set of $m$ data points $X_1, \ldots, X_m$, where to each point $X_i$, with $i \in \{1, 2, \ldots, m\}$, we associate a loss function $L_i(x)$, measuring the loss of any model $x$ at the data point $X_i$ (mapping the difference between the prediction of our model on $X_i$ with the true label to a real value), and we wish to identify a model $x^\star$ which minimizes the average loss over the dataset, that is, minimize the function  
$$ f ( x )= \frac{1}{m} \sum_{i = 1}^m L_i( x ).$$ 
This setting can be mapped to fundamental optimization tasks, such as linear regression over a given set of points, or training neural networks through backpropagation~\cite{rumelhart1986learning}. 
In both settings, we are given a large dataset, from which we pick a sample point uniformly at random in every iteration $t$. 
In every such iteration, a process reads the current version of the model $x_t$, computes the stochastic gradient $\tg(x_t)$ of this model with respect to the current sample by examining the derivative of the loss function at the point, given its label, and updates the model's components correspondingly. 
Thus, the stochastic gradients $\tg(x_t)$ correspond to the gradient of the model $x_t$ taken at a uniformly random data point $X_i$. Indeed, since the data point at which we compute the gradient is chosen uniformly at random, we have 
$$ \E[ \tg(x) ] = \frac{1}{m} \sum_{i = 1}^m \nabla L_i( x ) = \nabla f(x).$$

\paragraph{Parallel SGD.} The ability to parallelize this process at extremely large scale, up to thousands of nodes, e.g.~\cite{dean2012large}, has enabled researchers to approach new problems, and to reach super-human accuracy for several classic problems, e.g.~\cite{AlexNet, AlphaGo, he2016deep, seide2014sgd1bit, chilimbi2014project}. 
The standard way to distribute SGD to multiple compute nodes is \emph{data parallelism}:  given $n$ nodes, which we will abstract as parallel \emph{processes}, we split the dataset into $n$ partitions. Nodes process samples from their partitions in parallel, and synchronize by using a \emph{shared model} $x$. In this paper, we will focus on the case where the nodes are \emph{threads}, communicating through \emph{asynchronous shared memory}. 

In theory, data parallelism would allow the system to perform $n$ times more iterations (and model updates) per unit of time. The catch is that threads will need to \emph{synchronize} on the shared model, 
reducing scalability. 
In fact, early studies~\cite{langford2009slow} proposed using \emph{coarse-grained locking} to keep the process consistent to a sequential execution. As expected, coarse-grained locking leads to significant loss of performance. 

In this context,  a breakthrough by Niu et al.~\cite{recht2011hogwild} proved the unintuitive result that SGD should be able to converge even \emph{without coarse-grained  synchronization} to maintain model consistency, under some technical conditions, including high sparsity of the gradient updates, and under the assumption that individual updates  are applied via fetch-and-add synchronization operations. 
(This latter assumption appears necessary in general, since otherwise a delayed thread could completely obliterate all progress made up to some point, by overwriting the entire model, which resides in shared memory.) 
This work has sparked significant research on the convergence properties of asynchronous SGD, e.g.~\cite{desa2015hogwild, de2015taming, liu2015asynchronous, lian2015asynchronous, duchi2015asynchronous, zhang2017yellowfin}, improving theoretical bounds and considering  more general settings.  

In a nutshell, these results show, under various delay models and analytic assumptions, that SGD can still converge if iterations are asynchronous, that is, they cannot be ``serialized'' in any meaningful way. 
At the same time, all known convergence upper bounds, e.g.~\cite{de2015taming, lian2015asynchronous}, have a linear dependence in $\tau_{\max}$, the maximum delay between the time where a gradient update is \emph{generated} by any thread, and the point where the gradient is \emph{applied} to the model.\footnote{In shared-memory parlance, $\tau_{\max}$ is upper bounded by the \emph{interval contention}, where we define SGD iterations as individual operations. Intuitively, $\tau_{\max}$ is upper bounded by the number of iterations which can take steps between the start and end points of any fixed iteration.} 
This means that asynchronous SGD will take $\tau_{\max}$ times more iterations to converge, compared to the synchronous variant.  
Since $\tau_{\max}$ is technically only upper bounded by the \emph{length of the execution}, this upper bound appears extremely harsh. 
It is therefore natural to ask if this dependency is inherent, or whether it can be improved, yielding superior convergence rates for this classical algorithm. 

\paragraph{Contribution.} 
In this paper, we address this problem. 
Our approach is to express data-parallel SGD in the classic asynchronous shared memory model, against a strong (adaptive) adversarial scheduler, which designs schedules to delay the algorithm from converging, with full knowledge of the algorithm and random coin flips.
Under this formulation, our main results are as follows:

\begin{itemize}
    \item We show that, under standard analytic assumptions, for convex objective functions $f$, there exists a simple variant of the classic SGD algorithm which still converges under this strong adversarial model. 
    
    \item We prove that, under reasonable parameter settings, the convergence slowdown of SGD iterations due to asynchrony can be upper bounded by $\sqrt{\tau_{\max} n}$, where $\tau_{\max}$ is the \emph{maximum} interval contention over all operations and $n$ is the number of threads. 
    This result shows for the first time that the runtime dependence need not be linear in $\tau_{\max}$, even against a strong adversary, and is our main technical contribution. 
    \item We prove that, in general, the adversary can cause a convergence slowdown that is \emph{linear} in $\tau_{\max}$, if the algorithm does not decrease the step size $\alpha_t$ to offset the influence of stale gradient updates. This shows for the first time that the adversary can consistently and significantly slow down convergence by inducing delays, and that decreasing the step size (which is done by our algorithm) is in fact \emph{necessary} for good convergence under asynchrony.    

\end{itemize}

In sum, our results give new and improved upper and lower bounds on the ``price of asynchrony'' when executing the fundamental SGD algorithm in a concurrent setting. 
They show that this classic optimization tool can converge faster and with a wider range of parameters than previously known, under asynchronous iterations. 
At the same time, we exhibit a simple yet fundamental trade-off between the maximum delay in the system and the convergence rate of SGD, which governs the set of parameters under which SGD can still work efficiently. 

\paragraph{Techniques and Related Work.}
From the technical perspective, our results build upon martingale-based approaches for bounding the  convergence of SGD, e.g.~\cite{de2014global, de2015taming}. These techniques complement the classic ``regret'' bounds for characterizing the convergence of SGD, e.g.~\cite{bubeck2015convex}. 
We exploit martingale-based techniques in the asynchronous setting. To our knowledge, the only other work to employ such techniques for convex SGD is~\cite{de2015taming}, whose results we significantly extend. 
Specifically, with respect to this reference, the main departures are that 1) we consider a more challenging adaptive adversarial model, as opposed to a stochastic scheduling model; 2) eliminate the requirement that gradients contain a single non-zero entry, thereby significantly expanding the applicability of the framework; 3) reduce the linear convergence dependency in $\tau_{\max}$ to one of the form $\sqrt{\tau_{\max} n}$; 4) prove lower bounds on the slowdown due to asynchrony. 

There is an extremely vast literature studying the convergence properties of asynchronous optimization methods~\cite{recht2011hogwild, liu2015asynchronous, lian2015asynchronous, de2015taming, duchi2015asynchronous, leblond2018improved, nguyen2018sgd}, as well as efficient parallel implementations, e.g.~\cite{sallinen2016high, zhang2017yellowfin}, starting with seminal work by Bertsekas and Tsitsiklis~\cite{BT}. 
A complete survey is beyond the scope of this paper, and we therefore focus on work that is directly related to ours. 
Reference~\cite{recht2011hogwild} showed for the first time that, under strong analytical assumptions on sparsity and on the target loss function,  asynchronous SGD can still converge, and that, moreover, the convergence rate can be similar to that of the baseline under further assumptions on the parameters. 
Agarwal and Duchi~\cite{agarwalduchi} showed that, under strong ordering assumptions, delayed gradient computation does not affect the convergence of SGD.  
Lian et al.~\cite{lian2015asynchronous} provided ergodic convergence rates for asynchronous SGD for \emph{non-convex} objectives. 
Duchi et al.~\cite{duchi2015asynchronous} considered a model similar to ours, and showed that the impact of any asynchrony on the rate at which the algorithm converges is negligible, under strong technical assumptions on the convex function $f$ to be optimized, on the structure of its optimum, and on the sampling noise. 
Concurrent work~\cite{leblond2018improved, nguyen2018sgd} provides significantly more general analyses of iterative processes under asynchrony, covering several important optimization algorithms. 
With the exception of~\cite{duchi2015asynchronous}, which makes strong technical assumptions, all previous results for asynchronous SGD had a linear dependence in the maximum delay $\tau_{\max}$. We improve this dependency in this work.  

There exists significant work on mitigating the effects of asynchrony in applied settings, e.g.~\cite{zhang2016staleness, zheng2017asynchronous, DBLP:conf/allerton/MitliagkasZHR16}. A subset of these works are designed for a distributed shared memory setting, where it may be possible to examine the ``staleness'' of an update immediately before applying it, and adjust hyperparameters accordingly, and validated  empirically. 
By contrast, we consider an adversarial setting, where the scheduler actively attempts to thwart the algorithm's progress. Our lower bound applies to these works as well.

There has recently been significant work connecting machine learning and optimization with distributed computing. References~\cite{blanchard2017byzantine, yang2017byrdie, chen2017} consider distributed SGD in a \emph{Byzantine} adversarial setting, but in a message-passing system. 
In a series of papers~\cite{su2016multi, su2016fault}, Su and Vaidya have considered the problem of adding fault-tolerance to the problem of multi-agent optimization, as well as non-Bayesian optimization under asynchrony and crash failures~\cite{su2016asynchronous}. 

\section{Model}

\paragraph{Asynchronous Shared Memory.} We consider a standard asynchronous shared-memory model~\cite{AttiyaWelch},
in which $n$ threads (or processes) $P_1, \ldots, P_{n}$, communicate through atomic memory locations called registers, on which they perform atomic operations such as \lit{read}, \lit{write},  \lit{compare\&swap} and \lit{fetch\&add}. In particular, the algorithm we consider employs atomic \lit{read} and \lit{fetch\&add} operations, which are now standard in mass-produced multiprocessors. 
The \lit{fetch\&add} operation takes one argument, and returns the value of the register before the increment was performed, incrementing its value by the corresponding operand. 

As is usual, we will assume a \emph{sequentially consistent} memory model, in which once a thread returns from its invocation of a primitive (for example, a \lit{fetch\&add}), the value written by the  thread is immediately applied to shared memory, and henceforth visible by other processors. 

\paragraph{The Adversarial Scheduler.} Threads follow an algorithm  composed of
shared-memory steps and local computation, including random coin flips. 
The order of process steps is controlled by an adversarial entity we call the \emph{scheduler}. 
Time is measured in terms of the number of shared-memory steps scheduled by the adversary. 
The adversary may choose to crash a set of at most $n - 1$ processes by not scheduling them for the rest of the execution.  
A process that is not crashed at a certain step is \emph{correct}, and if it never crashes then it takes an infinite number of steps in the execution.
In the following, we assume a standard \emph{strong} adversarial scheduler, 
which can see the results of the threads' local coins when deciding the scheduling. 

\paragraph{Contention Bound.} In the following, fix a (concurrent) SGD iteration $\theta$, and let $\rho(\theta)$ be the \emph{interval contention} of the iteration $\theta$, defined as the number of SGD iterations which can execute concurrently with $\theta$. 
Let $\tau_{\max}$ be the upper bound over all these interval contention values, i.e.   $\tau_{\max} = \max_{\theta} \rho(\theta)$. 
Let $\tau_{avg}$ be an upper bound on the \emph{average} interval contention over all iterations during the (finite) execution of a program, i.e.  $\tau_{avg} = 1 / T \sum_{1 \leq \theta \leq T} \rho(\theta)$, where $T$ is the total number of iterations of the algorithm. 
It is known that $\tau_{avg} \leq 2n$, where $n$ is the number of threads~\cite{gibson2015non}.

\section{Background on Stochastic Gradient Descent}

Stochastic gradient descent (SGD) is a standard method in  optimization. Let $f: \R^d \to \R$ be a function which we want to minimize.
We have access to stochastic gradients $\tg$ of $f$, such that $\E [\tg(x)] = \nabla f(x)$. 
SGD will start at a randomly chosen point $x_0$, and converge towards the minimum by iterating the procedure 
\begin{equation}
x_{t + 1} = x_{t} - \alpha_t \tg(x_t), 
\label{eqn:sgd}
\end{equation}
\noindent where $x_t$ is the current candidate, and $\alpha_t$ is a variable step-size parameter.
Notably, this arises if we are given i.i.d. data points $X_1, \ldots, X_m$ generated from an unknown distribution $D$, and a loss function $\ell(X, \theta)$, which measures the loss of the model $\theta$ at data point $X$. We wish to find a model $\theta^*$ which  minimizes $f(\theta) = \E_{X \sim D} [\ell (X, \theta)]$, the expected loss to the data. 

Unless otherwise noted, we consider SGD for \textit{convex optimization} and with a \textit{constant learning rate}, that is $\alpha_t = \alpha$ for all $t$. We also make the following standard assumptions:
\begin{itemize}
    \item The function $f$ is strongly convex with parameter $c$, i.e. for all $x, y \in \mathbb{R}^d$:
    \begin{equation}
        \label{eqn:convex_assumption_f}
        \left(x-y\right)^{T}\left(\nabla f\left(x\right) - \nabla f \left(y\right)\right) \geq c\|x-y\|^2
    \end{equation}
    \item The stochastic gradient $\tg\left(x\right)$ is $L$-Lipschitz continuous in expectation for some $L>0$, i.e. for all $x, y\in \mathbb{R}^d$:
    \begin{equation}
        \label{eqn:lipschitz_assumption_f}
        \mathbb{E}\left(\|\tg\left(x\right) - \tg\left(y\right)\|\right) \leq L\|x-y\|
    \end{equation}
    \item The second moment of the stochastic gradients is bounded by some $M^2>0$, i.e. for all $x \in \mathbb{R}^d$: 
    \begin{equation}
        \label{eqn:grad_is_bounded_assumption_f}
        \mathbb{E}\left(\|\tg\left(x\right)\|^2\right) \leq M^2.
    \end{equation}
\end{itemize}
Classic approaches for analyzing the convergence of SGD attempt to bound the distance between the expected value of $f$ at the average of the currently generated iterates and the optimal value of the function (e.g. Theorem 6.3 in~\cite{bubeck2015convex}), showing that this distance decreases linearly with the number of iterations.

Here we consider a different approach that aims at estimating the probability that the algorithm has failed to converge to a success region around the optimal parameter value after $T$ steps. To this end, we employ a martingale-based analysis of the algorithm, an approach that has recently become a popular tool for analyzing asynchronous optimization algorithms~\cite{chaturapruek2015asynchronous,de2015taming}. Let $x^*$ be the minimizer of the function $f$. Given an $\epsilon > 0$, we denote by $$S = \{x\text{ }|\text{ }\|x-x^{*}\|^2 \leq \epsilon\}$$ the success region around this minimizer, to which we want to converge. Our analysis aims to bound the probability of the event $F_T$ that $x_i \not\in S$ for all $i \leq T$, i.e. the event that the algorithm has failed to hit the success region by time $T$.

An existing result of this type about the convergence of parallel SGD was derived in \cite{de2015taming}. Under the non-standard additional assumption that each gradient update on the parameter value only effects a single entry of $x_t$ (i.e., that the stochastic gradients contain a single non-zero entry),\footnote{Our analysis eliminates this assumption, a result which may be of independent interest.} one can obtain the following:
\begin{theorem}[\cite{de2015taming}]
\label{thm:sgd_result_basic}
Consider the SGD algorithm under the assumptions above, run for $T$ steps, with a success region $S = \{x\text{ }|\text{ }\|x-x^{*}\|^2 \leq \epsilon\}$ and with learning rate $\alpha = \frac{c\epsilon \vartheta}{M^2}$ for some constant $\vartheta \in \left(0,1\right)$. Then the probability of the event $F_T$ that $x_i \not\in S$ for all $i \leq T$ is:
\begin{equation}
    \label{eqn:conv_rate_sgd_basic}
    \mathbb{P}\left(F_T\right) \leq \frac{M^2}{c^2\epsilon\vartheta T}\log\left(e \|x_0 - x^{*}\|^2\epsilon^{-1}\right).
\end{equation}
\end{theorem}
\noindent Note that this bound also decreases linearly with the number of iterations.
\section{Lock-Free SGD in Shared-Memory}

A standard way of parallelizing the SGD algorithm is to have multiple parallel threads execute the procedure in Equation~\ref{eqn:sgd}. 
We assume a lock-free setting, in which threads share the set of parameters (model) $X[d]$, which they can read and update entry-wise (through \lit{read} and \lit{fetch\&add} operations) concurrently, without additional synchronization. Each thread executes the steps in Algorithm \ref{algo:sgd}:

\begin{minipage}[c]{0.49\textwidth}
  \vspace{0pt}
\begin{algorithm}[H]
\DontPrintSemicolon 
 \SetAlgoLined
 {\small
\KwIn{Dataset $D$, dimension $d$, function $f$}
\KwOut{Minimizer $X$, initially $X = (0, 0, \ldots, 0)$}
\KwData{Shared copy of the parameter array $X$, of dimension $d$ }
\KwData{Iteration counter $C$, learning rate $\alpha$ }

\noindent \textbf{procedure} $\mathsf{EpochSGD}(T, \alpha)$\;
\Indp
\For{each iteration $\theta$} {
\lIf{ $ C.\fadd{}(1) \geq T$ } {\textbf{ return } }
\lFor{$j$ from $1$ to $d$} {
	$v_{\theta}[j] \gets X[j].\textnormal{read}()$
}
	$\tg_{\theta} \gets $  stochastic gradient at $v_{\theta}$, so $\mathbb{E}\left(\tg_{\theta}\right) = f\left(v_{\theta}\right)$\;
	\For{ $j$ from $1$ to $d$ } {
		\lIf{ $\tg_{\theta}[j] \neq 0$ } {
		$X[j].\fadd( - \alpha \tg_{\theta}[j] )$
		}
	}
	
\Indm
}
}
\caption{{\sc SGD} code for a single thread.}
\label{algo:sgd}
\end{algorithm}
\end{minipage}~~~~
\begin{minipage}[c]{0.44\textwidth}
\centering
\includegraphics[scale=.35]{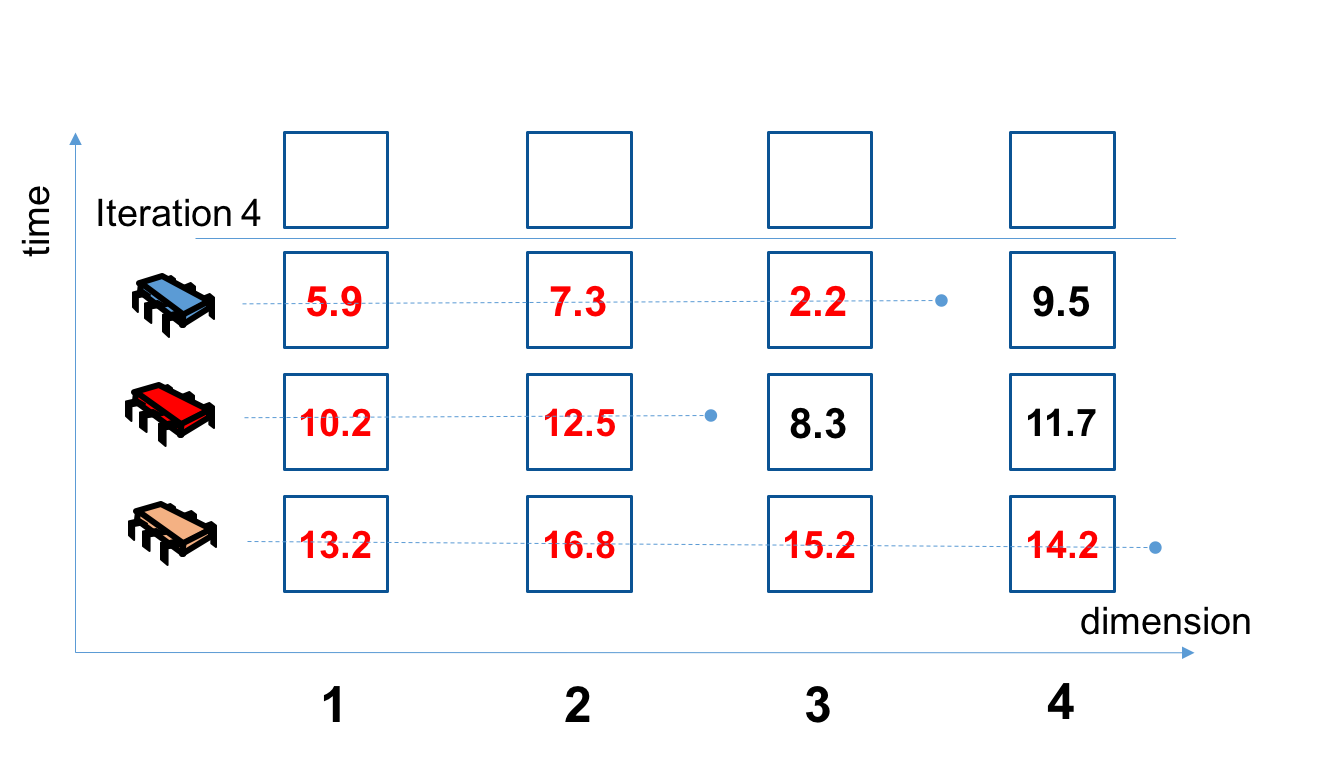}
\captionof{figure}{Algorithm modeling. The updates in red have been applied to shared memory, while the updates in black are still pending. The updates in black or red form $x_t$. The dot is the point where the thread has currently stopped updating. To obtain the value of $v_t$ at the current time, we sum values in red on each column (i.e. all values up to the dot on the row).}
\label{fig:illustration}
\end{minipage}
\vspace{1em}

We emphasize that this modeling of the algorithm is standard: virtually all papers which analyze asynchronous SGD consider this formulation, e.g.~\cite{recht2011hogwild, de2015taming, duchi2015asynchronous}. Updates are assumed to occur via \fadd{}, to avoid complete resets of the state by a delayed thread. 

\section{A Slowdown Lower Bound via Adversarial Delays}

\paragraph{Setup.} We now provide a simple argument which yields a lower bound on the achievable speedup if the adversary can delay gradients by a large $\tm$. This argument might also serve as a brisk hands-on introduction to SGD. 
We consider a standard setting, where two threads each have access to local gradient samples, and share the model.  
Assume we are trying to minimize the (convex) objective function
\[
  f(x) = \frac{1}{2} x^2,
\]
and have access to its noisy gradients
\[ \tg(x) = x - \tilde u,
\]
where $\tilde u$ is normally distributed noise with mean $0$ and variance $\sigma^2$. Observe the minimum is at $0$, and $\E [\tg(x) ] = \nabla f(x)$.
Now if $u_t$ is the stochastic gradient at time $t$, the SGD update is
\[
  x_{t+1} = x_t - \alpha \tg(x_t) = (1 - \alpha) x_t + \alpha \tilde u_t.
\]
\paragraph{Adversarial Strategy.} Suppose that the adversary is executing the following strategy. 
First, both threads generate a gradient with respect to $x_0$. The first thread then executes for $\tau$ consecutive iterations, and then the second writes the state with a gradient from the initial value. Let us now analyze the rate at which this algorithm can converge. 

\paragraph{Analysis.} After the first thread runs, the output will be
\[
  x_{\tau} = (1 - \alpha)^{\tau} x_0 + \alpha \sum_{k = 1}^{\tau} (1 - \alpha)^{\tau - k} \tilde u_{k-1}. 
\]
\noindent After the second thread merges in its stale gradient, which has the form $x_0 - \tilde u_{\tau}$, we have:
\begin{align*}
  x_{\tau + 1}
  = 
  \left( (1 - \alpha)^{\tau} - \alpha \right) x_0
  +
  \alpha \sum_{k = 1}^{\tau} (1 - \alpha)^{\tau - k} \tilde u_{k-1} + \alpha \tilde u_{\tau}.
\end{align*}
The second term of this last expression will be a zero-mean Gaussian with a variance of
\begin{align*}
  \bar \sigma^2
  &=
  \alpha^2 \sum_{k = 1}^{\tau} (1 - \alpha)^{2\tau - 2k} \sigma^2 + \alpha^2 \sigma^2 
  \\ &=
  \alpha^2 \sigma^2 \left( 1 + \sum_{k = 0}^{\tau - 1} \left( (1 - \alpha)^2 \right)^{k} \right)  
  \\ &=
  \alpha^2 \sigma^2 \left( 1 + \frac{1 - (1 - \alpha)^{2\tau}}{1 - (1 - \alpha)^2} \right).
\end{align*}
For a fixed $\alpha$, if we choose $\tau$ large enough that $2 (1 - \alpha)^{\tau} \le \alpha$, and suppose for simplicity that $\sigma = 0$, then we can get
\[
  \norm{x_{\tau+1}} \ge \frac{\alpha}{2} \norm{x_0}, \textnormal{ versus } \norm{x_{\tau+1}} = (1 - \alpha)^{\tau} \norm{x_0} 
\]

\noindent in the case with no adversary. To compare these rates, we take the logarithm, and obtain a slowdown factor of
\[
  \frac{\log\left( (1 - \alpha)^{\tau} \right)}{ \log\left( \frac{\alpha}{2} \right) }
  =
  \tau
  \frac{\log(1 - \alpha)}{ \log(\alpha) - \log(2) }
  =
  \Omega(\tau),
\]
which implies an $\Omega(\tau)$ factor slowdown is possible from a delay of $\tau$.
This shows that with a maximum delay $\tau_{\max} = \tau$, the adversary can achieve an asymptotic slowdown that is linear in $\tau$. We conclude as follows.
\begin{theorem}
    \label{thm:lb}
    Given an instance of the lock-free SGD Algorithm in~\ref{algo:sgd} with fixed learning rate $\alpha$, there exists an adversarial strategy with maximum delay $\tau_{\max} = O( \log (\alpha ) / \log ( 1 - \alpha ))$ such that the algorithm converges $\tau_{\max}$ times slower than the sequential variant. 
\end{theorem}

\section{Convergence Upper Bounds in Asynchronous Shared Memory}

\subsection{Preliminaries}

\paragraph{Notation and Iteration Ordering.} We now introduce some notation, and a few basic claims about the above concurrent process. 
First, we define an order on the above iterations, performed possibly by distinct threads, by the time at which the iteration performs its first \fadd{} operation, on the first model component $X[1]$. (Here we are using the sequential consistency property of the memory model.) 
This ordering induces a useful total order between iterations: for any integer $t \geq 1$, iteration $t$ is the $t$th iteration to complete its \fadd{} on $X[1]$. 
We now note that all of these iterations up to iteration $t$ must have completed their computation, but may not have completed \emph{writing their updates} to $X$ by the time when iteration $t + 1$ starts. At most $n$ of these iterations may be incomplete at any given time. We formalize this as follows.

\begin{lemma}
\label{lem:order}
	Let iteration $t$ be the $t$th iteration to update $X[1]$.  This is a total order on the iterations. We say that an iteration is \emph{incomplete} at a given point in the execution if it has performed its first update (on $X[1]$), but has not completed its last update (on $X[d]$). For any $t \geq 1$, at most $n$ iterations with indices $\leq t$ can be incomplete.   
\end{lemma}

\noindent By assumption the maximal interval contention is bounded by $\tau_{\max}$. However, since at most $n$ threads can run at the same time, it is intuitive that the average contention should be $\mathcal{O}\left(n\right)$. We formalize this via the following:

\begin{lemma}
\label{lem:contention}
    Fix a parameter $K$, and an arbitrary time interval $I$ during which exactly 
    $Kn$ consecutive SGD iterations start. 
    We call an SGD iteration $\theta$ \emph{bad} if more than $Kn$ iterations start between its start time and end time. Otherwise, an SGD iteration is \emph{good}. 
    Then, the number of bad iterations which complete during $I$ is less than $n$. 
\end{lemma}
\begin{proof}
    Assume for contradiction that the number of bad iterations is $n$ or larger. 
    Then, there must exist a thread $p$ which completes two of these bad iterations during this interval. Denote the second iteration by $\theta$. 
    Since the two iterations by $p$ cannot be concurrent with each other, 
    the maximum number of iterations which can be concurrent with $\theta$ (except itself) is at most $Kn - 2$. 
    Hence, $\theta$ cannot be bad, a contradiction. 
\end{proof}

\paragraph{Inconsistent Views $v_t$ and Accumulators $x_t$.}
Let $\tg_{t}$ be the gradient vector which the corresponding thread executing iteration $t$ wishes to add to $X$, and $v_t$ be this thread's \emph{view} of the model $X$ when generating this gradient, composed out of the values the thread read for each entry. 
Importantly, we define the auxiliary (global) vector $x_t =\sum_{k = 1}^t \tg_k$, containing all of the updates which various \emph{started} iterations wish to apply to the model $X$ up to $t$. We note that the view $v_t$ which iteration $t$ uses to generate its gradient update can consist of a possibly inconsistent set of updates across the various components of $X$. However, at all times 1) all of these updates must be contained in $x_t$, in the sense that they feature in $\tg_k$ for some $k\leq t$; and 2) if we denote by $t'$ the highest iteration whose update is read by $v_t$, we must have that there can be at most $\ell(\theta)$ updates with indices $< t'$ which $v_t$ sees as incomplete during its scan. This is illustrated in Figure \ref{fig:illustration}.

\paragraph{SGD Convergence.} We now analyze the convergence of SGD under this lock-free model. Following \cite{de2015taming}, we employ a martingale approach for proving convergence rates of SGD.

Given a probability space $\left(\Omega, \mathcal{F}, \mathbb{R}\right)$, a real-valued stochastic process $W:\mathbb{N}\times\Omega \rightarrow \mathbb{R}$ is called a supermartingale if $\mathbb{E}\left(W_{t+1}|W_{t}, W_{t-1}, ..., W_0\right) \leq W_{t}$ for all $t \geq 0$. Here we consider supermartingales of the form $W_t\left(x_t, ..., x_0\right)$, defined based on the state of the optimization algorithm. Intuitively, $W_t$ represents our unhappiness about the state of the algorithm at time $t$. More precisely, we have the following:
\begin{definition}[\cite{de2015taming}]
Given a stochastic optimization algorithm, a non-negative process $W_t:\mathbb{R}^{d\times t} \rightarrow \mathbb{R}$ is a rate supermartingale with horizon $B$, if two conditions hold. Firstly, it is a supermartingale, i.e. for all $x_t, ..., x_0$ and for all $t\leq B$:
\begin{equation}
\label{eqn:supermartingale_defn}
    \mathbb{E}\left(W_{t+1}\left(x_t - \alpha \tg\left(x_t\right), x_{t}, ..., x_0\right)\right) \leq W_{t}\left(x_t, x_{t-1}, ..., x_0\right),
\end{equation}
where expectation is taken with respect to the randomness at time $t$ and conditional on the past. Secondly, for any time $T\leq B$ and any sequence $x_T, ..., x_0$, if the algorithm has not succeeded by time $T$ (i.e. $x_t\not\in S$ for all $t \leq T$), then:
\begin{equation}
    \label{eqn:rate_property}
    W_T\left(x_T, x_{T-1}, ..., x_0\right) \geq T.
\end{equation}
\end{definition}
\noindent The main result in~\cite{de2015taming} shows how constructing such a supermartingale for an optimization algorithm can be used to obtain a bound on the probability that the algorithm has not visited the success region after a certain number of iterations. Under the considered stochastic scheduling model, the authors employ a parameter $\tau$, that denotes the \emph{worst-case expected delay} caused by the parallel updates. Under the additional assumption that the stochastic gradients contain a single non-zero entry, the following result is derived:
\begin{theorem}[\cite{de2015taming}]
\label{thm:sgd_result_buckwild}
Consider the SGD algorithm for optimizing a convex function $f$ that satisfies the assumptions above and under the asynchronous model of~\cite{de2015taming}, with a success region $S = \{x\text{ }|\text{ }\|x-x^{*}\|^2 \leq \epsilon\}$ and with learning rate $\alpha = \frac{c\epsilon \vartheta}{M^2 + 2LM\tau\sqrt{\epsilon}}$ for some constant $\vartheta \in \left(0,1\right)$. Then the probability of the event $F_T$ that $x_i \not\in S$ for all $i \leq T$ is:
\begin{equation}
    \label{eqn:conv_rate_sgd_buckwild}
    \mathbb{P}\left(F_T\right) \leq \frac{M^2 + 2LM\tau \sqrt{\epsilon}}{c^2\epsilon\vartheta T}\log\left(e \|x_0 - x^{*}\|^2\epsilon^{-1}\right).
\end{equation}
\end{theorem}
Again, the bound decreases linearly with the number of iterations. However, there is an additive term that increases with $\tau$. In particular, the bound on the failure probability is worse than in the sequential SGD case described previously.
\subsection{Convergence Analysis}
We now apply a martingale analysis similar to the one in \cite{de2015taming} to obtain results about the rate of convergence of SGD under the Asynchronous Shared Memory model. We denote the maximum delay at time $t$ by $\tau_t \geq 0$. We also assume that all $\tau_t$ are bounded by some maximum $\tau_{\textrm{max}}$, i.e. $\tau_t \leq \tau_{\max}$ for all $t$.

Note that at any time $t$, the gradient is computed based on a view $v_t$ that might be missing updates from only the last $\tau_t$ iterations. Therefore,
\begin{align*}
\|x_t - v_t\|_1 & \leq \sum_{k=1}^{\tau_t}\|x_{t-k+1} - x_{t-k}\|_1 \\ & \leq \sum_{k=1}^{t}\|x_{t-k+1} - x_{t-k}\|_1\mathbbm{1}_{\{\tau_t \geq k\}}
\end{align*}
Now since for any $x \in \mathbb{R}^d$, we have $\|x\|_2 \leq \|x\|_1 \leq \sqrt{d}\|x\|_2$, it follows that:
\begin{align}
\label{bound_on_xt_vt}
\|x_t - v_t\| \leq \sqrt{d}\sum_{k=1}^{t}\|x_{t-k+1} - x_{t-k}\|\mathbbm{1}_{\{\tau_t \geq k\}}
\end{align}
For the subsequent analysis we will also need the following:
\begin{lemma}
\label{lem:lemma_on_tau_t}
For any $t$:
\begin{align}
\label{eqn:bound_on_indicators}
\sum_{m=1}^{\infty} \mathbbm{1}_{\{\tau_{t+m} \geq m\}} = \sum_{m=1}^{\tau_{\textrm{max}}} \mathbbm{1}_{\{\tau_{t+m} \geq m\}} \leq 2\sqrt{\tau_{\textrm{max}}n}.
\end{align}
\end{lemma}
\begin{proof}
By Lemma \ref{lem:contention}, we know that for any constant $K$ and for any $Kn$ consecutive steps $t+1, ..., t+Kn$, $\tau_{t+i} > Kn$ for at most $n$ indexes. Hence,
\begin{align*}
    \sum_{m=1}^{\tau_{\textrm{max}}} \mathbbm{1}_{\{\tau_{t+m} \geq m\}}  & \leq Kn + \sum_{m=Kn + 1}^{\tau_{\textrm{max}}} \mathbbm{1}_{\{\tau_{t+m} > Kn\}}  \\ & \leq Kn + \left(\frac{\tau_{\textrm{max}} - Kn}{Kn} + 1\right)n \\ & = \frac{\tau_{\textrm{max}}}{K} + Kn.
\end{align*}
This holds for any positive $K$. The bound is minimized at $K = \sqrt{\frac{\tau_{\textrm{max}}}{n}}$, which yields the result.
\end{proof}
Next, we obtain a bound on the probability that the algorithm has not visited a given success region $S = \{x\text{ }|\text{ }\|x-x^{*}\|^2 \leq \epsilon\}$. To this end, we show a result similar to the one in Theorem 1 in \cite{de2015taming}. We will assume the existence of a rate supermartingale with respect to the underlying sequential SGD process that is Lipschitz in its first coordinate and show that this can be used to obtain a bound on the failure probability. The exact assumptions on $W$ are as follows:
\begin{itemize}
    \item $W$ is a supermartingale with horizon $B$ with respect to the sequential SGD process $x_{t+1} = x_t - \alpha \tg\left(x_t\right)$. Note that $W$ need not be a supermartingale with respect to the lock-free SGD algorithm.
    \item For any $T > 0$, if $x_i \not\in S$ for all $i \leq T$, then:
    $W_T\left(x_T, \ldots, x_0\right) \geq T.
    $ Otherwise, we say that the algorithm \textit{has succeeded} at time $T$.
    \item $W$ is Lipschitz continuous in the current iterate with parameter $H$, i.e. for all $t, u, v$ and any $x_{t-1}, ..., x_0$: $$\|W_t\left(u, x_{t-1}, ..., x_0\right) - W_t\left(v, x_{t-1}, ..., x_0\right)\| \leq H\|u-v\|$$
\end{itemize}
Under these assumptions, we can prove our main technical claim, whose proof is deferred to the Appendix. 
\begin{theorem}
\label{thm:bound_on_fail_prob}
Assume that $W$ is a rate supermartingale with horizon $B$ for the sequential SGD algorithm and that $W$ is $H$-Lipschitz in the first coordinate. Assume further that $\alpha^2HLMC\sqrt{d} < 1$, where $C = 2\sqrt{\tau_{\textrm{max}}n}$. Then for any $T\leq B$, the probability that the lock-free SGD algorithm has not succeeded at time $T$ (that is, the probability of the event $F_T$ that $x_i \not\in S$ for all $i \leq T$) is:
\begin{equation}
\label{bound_on_prob}
    \mathbb{P}\left(F_T\right) \leq \frac{\mathbb{E}\left(W_0\left(x_0\right)\right)}{\left(1 - \alpha^2HLMC\sqrt{d}\right)T}.
\end{equation}
\end{theorem}

\noindent We now apply the result with a particular choice for the martingale $W_t$. We use the process proposed in~\cite{de2015taming} in the case of convex optimization, which they show is a rate supermartingale for the sequential SGD process. More precisely, we have the following:
\begin{lemma}[\cite{de2015taming}]
  \label{lemmaConvexW}
  Define the piecewise logarithm function to be
  \[
    \textnormal{plog}(x)
    =
    \left\{
      \begin{array}{lr}
        \log(e x) & : x \ge 1 \\
        x & : x \le 1
      \end{array}
    \right.
  \]
  Define the process $W_t$ by:
  $$
    W_t(x_t, \ldots, x_0)
    =
    \frac{
      \epsilon
    }{
      2 \alpha c \epsilon
      -
      \alpha^2 M^2
    }
    \textnormal{plog}\left(\|x_t - x^*\|^2 \epsilon^{-1} \right)
    +
    t.
  $$\\
  If the algorithm has not succeeded by timestep
  $t$ (i.e. $x_i \not\in S$ for all $i \leq t$) and by $W_t = W_{u-1}$ whenever $x_i \in S$ for some $i \leq t$ and $u$ is the minimal index with this property. Then $W_t$ is a rate supermartingale for sequential SGD
  with horizon $B = \infty$. It is also $H$-Lipschitz in the first coordinate, with $H = 2\sqrt{\epsilon}\left(2\alpha c\epsilon - \alpha^2 M^2\right)^{-1}$, that is for any $t, u, v$ and any sequence $x_{t-1}, \ldots, x_0$: $$\|W_t\left(u, x_{t-1}, \ldots, x_0\right) - W_t\left(v, x_{t-1}, \ldots, x_0\right)\| \leq H\|u-v\|.$$
\end{lemma}
\noindent Using this martingale in Theorem \ref{thm:bound_on_fail_prob}, we obtain the following result, whose proof is deferred to the Appendix:
\begin{corollary}
\label{cor:exact_bound}
Assume that we run the lock-free SGD algorithm under the Asynchronous Shared Memory model for minimizing a convex function $f$ satisfying the listed assumptions. Set the learning rate to:
\begin{equation}
\label{eqn:chosen_learning_rate}
\alpha = \frac{c\epsilon\vartheta}{M^2 + 2\sqrt{\epsilon}LMC\sqrt{d}} = \frac{c\epsilon\vartheta}{M^2 + 4\sqrt{\epsilon}LM\sqrt{\tau_{\textrm{max}}n}\sqrt{d}},
\end{equation}
for some constant $\vartheta\in (0,1]$. Then for any $T > 0$ the probability that $x_i \not\in S$ for all $i \leq T$ is:
\begin{equation}
\label{eqn:exact_bound}
\mathbb{P}\left(F_T\right) \leq \frac{M^2 + 4\sqrt{\epsilon}LM\sqrt{\tau_{\textrm{max}}n}\sqrt{d}}{c^2 \epsilon\vartheta T}\textnormal{plog}\left(\frac{e\|x_0 - x^*\|^2}{\epsilon}\right).
\end{equation}
\end{corollary}

\noindent Choosing $\vartheta = 1$ gives the smallest value for the upper bound under this setting. We can impose an arbitrary small learning rate by selecting a small  $\vartheta$, while still ensuring convergence.

\section{Iterated Algorithm with Guaranteed Convergence}

The procedure in Algorithm \ref{algo:sgd} has the property that, for an appropriately chosen stopping time $T$, it will eventually reach the ``success'' region, where the distance to the optimal parameter value falls below $\epsilon$. However, due to asynchrony and adversarial updates, threads might perform updates which cause them to \emph{leave} the success region: a  delayed thread might apply stale gradients to the model, overwriting the progress. 

We now present an extension of the algorithm which deals with this problem, and allows us to converge to a success region of radius $\epsilon$ around the optimum $x^*$, for any $\epsilon > 0$. The algorithm will run a series of \emph{epochs}, each of which is a series of $T$ SGD iterations, executed using the procedure in Algorithm~\ref{algo:sgd}. The only difference between epochs is that they are executed with an exponentially decreasing learning rate $\alpha$. (We note that this epoch pattern is already used in many settings, such as neural network training.)   
The epochs share the model $X$, with the critical note that we require that a gradient update can only be applied to $X$ in the same epoch when it was generated. This condition can be enforced either by maintaining an epoch counter, on which threads condition their update via double-compare-single-swap (DCAS), or by having a distinct model allocated for each epoch.

\begin{algorithm}
\DontPrintSemicolon 
 \SetAlgoLined
 {\small
\KwIn{Dataset $D$, dimension $d$, function $f$ to minimize}
\KwOut{Minimizer $X$, initially $X = (0, 0, \ldots, 0)$}
\KwData{Learning rate $\alpha$, iteration count $T$, gradient bound $M$}

 \CommentSty{//At thread $i$:}\;
 \noindent \textbf{procedure} $ \mathsf{FullSGD}(\epsilon, \alpha, n) $ \;
 \Indp
 \For{ $j$ from $0$ to $\log \left( \alpha {  2M n / \sqrt \epsilon } \right) $} {
     $ \mathsf{EpochSGD}( T, \alpha ) $ \;
      $ \alpha \gets \alpha / 2 $ \;
 }
 \CommentSty{//The last epoch:} \;
 Execute $ \mathsf{EpochSGD}( T, \alpha ) $, accumulating local gradients into $ Acc[i] $ \; 
 $ {r} \gets \mathsf{sum}_{k = 1}^n ( Acc[k] ) $ \CommentSty{//Collect entrywise sum}\;
 \textbf{return} ${r}$\;

 \Indm
}

\caption{Full {\sc SGD} code for a single thread.}
\label{algo:sgd_iterate}
\end{algorithm}

The only distinct epoch is the last, in which threads each aggregate the gradients they produced locally. At the end of the epoch, the threads will collect all local gradients locally. As we show in the Appendix, the model will be guaranteed to be close to optimal in expectation. 

\begin{corollary}
\label{cor:full}
    The \lit{FullSGD} procedure given in Algorithm~\ref{algo:sgd_iterate} guarantees that $\E [ \| r -  x^\star \| ] \leq \epsilon$ after executing for $O( T \log \left( \alpha {  2M n / \sqrt \epsilon } \right) )$ iterations. 
\end{corollary}

\noindent We can characterize the probability of success by reducing the target $\epsilon$, and applying Markov's inequality. 

\section{Discussion}

\paragraph{Learning Rates versus Asynchronous Convergence.} 
An immediate consequence of our lower bound in Theorem~\ref{thm:lb} is that the algorithm must either have a low initial learning rate, or lower the learning rate across multiple iterations (as in Algorithm~\ref{algo:sgd_iterate}) in order to be able to withstand adversarial delays. Otherwise, the adversary can always apply stale gradients generated at a far enough time in the past to nullify progress. 
An alternative approach, which we did not consider here, would be to introduce a ``momentum'' term by which the current model value is multiplied~\cite{DBLP:conf/allerton/MitliagkasZHR16}.  

\paragraph{Lower Bounds versus Upper Bounds.} The attentive reader might find it curious that our lower bound suggests a linear slowdown in $\tau_{\max}$, whereas our upper bound suggests that the slowdown is linear in $O( \sqrt{ \tm{} n })$. However, a close examination of the preconditions for these two results will reveal that they are in fact complementary: in the lower bound, given fixed learning rate $\alpha$, the adversary needs to set a large delay $\tm{} \geq (\log (\alpha / 2))/\log ( 1 - \alpha )$ in order to slow down convergence. 
At the same time, the upper bound in Theorem~\ref{thm:bound_on_fail_prob} requires that $2\alpha^2HLM\sqrt{d \tm n} < 1$, which is incompatible with the above condition. 
Specifically, our improved convergence bound  shows that asynchronous SGD converges faster and for a wider range of parameters than previously known. 

\paragraph{Why is Asynchronous SGD Fast in Practice?} In a nutshell, Theorem~\ref{thm:bound_on_fail_prob} shows that the gap in convergence between asynchronous SGD and the sequential variant becomes negligible if  $\alpha^2HLMC\sqrt{d \tm n} \ll 1$. 
Intuitively, this condition holds in practice since gradients are often sparse, meaning that $d$ is low, the delay factors $\tau_{\max}$ and $\tau_{avg}$ are \emph{not} set adversarially, and the learning rate $\alpha$ can be set by the user to be small enough to offset any increase in the other terms. In particular, $\tau_{\max}$ is limited  by the staleness of updates in the write buffer at each core, which is well bounded in practice~\cite{morrison2015temporally}. 

At the same time, it is important to note that, while we ``sequentialize'' iterations in our analysis, up to $n$ iterations may happen in parallel at any time, reducing the wall-clock convergence time by up to a factor of $n$. Thus, the practical trade-off is between any slow-down caused by asynchrony, and the parallelism due to multiple computation threads.

\section{Acknowledgements}
We would like to thank Martin Jaggi for useful discussions. 
This project has received funding from the European Union’s Horizon 2020 research and innovation programme under the Marie Skłodowska-Curie Grant Agreement No. 665385.

\bibliographystyle{plain}
\bibliography{references}

\appendix

\section{Deferred Proofs}

\subsection{Proof of Theorem~\ref{thm:bound_on_fail_prob}}

\begin{reptheorem}{thm:bound_on_fail_prob}
Assume that $W$ is a rate supermartingale with horizon $B$ for the sequential SGD algorithm and that $W$ is $H$-Lipschitz in the first coordinate. Assume further that $\alpha^2HLMC\sqrt{d} < 1$, where $C = 2\sqrt{\tau_{\textrm{max}}n}$. Then for any $T\leq B$, the probability that the lock-free SGD algorithm has not succeeded at time $T$ (that is, $x_i \not\in S$ for all $i \leq T$) is:
\begin{equation}
    \mathbb{P}\left(F_T\right) \leq \frac{\mathbb{E}\left(W_0\left(x_0\right)\right)}{\left(1 - \alpha^2HLMC\sqrt{d}\right)T}.
\end{equation}
\end{reptheorem}
\begin{proof}
Consider a process defined by $V_0 = W_0$ and by:
\begin{equation}
\label{supermart}
V_t = W_t - \alpha^2HLMC\sqrt{d}t + \alpha HL\sqrt{d}\sum_{k=1}^{t} \|x_{t-k+1} - x_{t-k}\|\sum_{m=k}^{\infty} \mathbbm{1}_{\{\tau_{t-k+m} \geq m\}},
\end{equation}
whenever $x_i \not\in S$ for all $0 \leq i \leq t$. Finally, if $x_i \in S$ for some $0 \leq i \leq t$, then define $$V_t\left(x_t, \ldots, x_{u-1}, \ldots, x_0\right) = V_{u-1}\left(x_{u-1}, \ldots, x_0\right),$$ where $u$ is the minimal index, such that $x_u \in S$.
\noindent Assume that the algorithm has not succeeded at time $t$. Using the Lipschitzness of $W$:
\begin{align*}
V_{t+1}\left(x_{t} - \alpha \tg\left(v_t\right), x_{t}, ..., x_0\right) & = W_{t+1}\left(x_t - \alpha \tg\left(v_t\right), x_{t}, ..., x_0\right) - \alpha^2 HLMC\sqrt{d}\left(t+1\right) \\ & + \alpha HL\sqrt{d}\|x_{t+1} - x_t\|\sum_{m=1}^{\infty}\mathbbm{1}_{\{\tau_{t + m} \geq m\}} \\ & + \alpha HL\sqrt{d}\sum_{k=2}^{t+1}\|x_{t-k+2} - x_{t-k+1}\|\sum_{m=k}^{\infty}\mathbbm{1}_{\{\tau_{t+1-k+m} \geq m\}} \\ & \leq W_{t+1}\left(x_t - \alpha \tg\left(x_t\right), \ldots, x_0\right) + \alpha H\|\tg \left(x_t\right) - \tg \left(v_t\right)\| \\ & - \alpha^2 HLMC\sqrt{d}\left(t+1\right) + \alpha HL\sqrt{d}\|x_{t+1} - x_t\|\sum_{m=1}^{\infty}\mathbbm{1}_{\{\tau_{t + m} \geq m\}} \\ & +  \alpha HL\sqrt{d}\sum_{k=1}^{t}\|x_{t-k+1} - x_{t-k}\|\sum_{m=k+1}^{\infty}\mathbbm{1}_{\{\tau_{t-k+m} \geq m\}}
\end{align*}
Now we take expectation with respect to the gradient at time $t$ and condition on the past (denoted by $\mathbb{E}_{t|.}$). Using the Lipschitzness of $\tg$ and the supermartingale property of $W_t$ for the first step and the bound on the expected norm of the gradient and Lemma \ref{lem:lemma_on_tau_t} for the second: 
\begin{align*}
\mathbb{E}_{t|.}\left(V_{t+1}\right) & \leq W_t\left(x_t, \ldots, x_0\right) + \alpha HL\|x_t - v_t\| - \alpha^2 HLMC\sqrt{d}t \\ & + \left(\alpha HL\sqrt{d}\mathbb{E}_{t|.}\left(\|\alpha \tg\left(v_t\right)\|\right)\sum_{m=1}^{\infty}\mathbbm{1}_{\{\tau_{t+m} \geq m\}} - \alpha^2 HLMC\sqrt{d}\right) \\ & + \alpha HL\sqrt{d}\sum_{k=1}^{t}\|x_{t-k+1} - x_{t-k}\|\sum_{m=k+1}^{\infty}\mathbbm{1}_{\{\tau_{t-k+m} \geq m\}} \\ & \leq W_t - \alpha^2 HLMC\sqrt{d}t + \alpha HL\sqrt{d}\sum_{k=1}^{t} \|x_{t-k+1} - x_{t-k}\|\sum_{m=k}^{\infty} \mathbbm{1}_{\{\tau_{t-k+m} \geq m\}} \\ & + \alpha HL\left(\|x_t - v_t\| - \sqrt{d}\sum_{k=1}^{t}\|x_{t-k+1} - x_{t-k}\|\mathbbm{1}_{\{\tau_{t}\geq k\}}\right) \leq V_{t}
\end{align*}
This inequality also trivially holds in the case when the algorithm has succeeded at time $t$. Hence, the process $V_t$ is a supermartingale for the lock-free SGD process.

\noindent Note also that if the algorithm has not succeeded at time $T$, then $W_T \geq T$ and hence:
\begin{align*}
V_T & = W_T - \alpha^2 HLMC\sqrt{d}T + \alpha HL\sqrt{d}\sum_{k=1}^{T} \|x_{T-k+1} - x_{T-k}\|\sum_{m=k}^{\infty} \mathbbm{1}_{\{\tau_{T-k+m} \geq m\}} \\ & \geq W_T - \alpha^2 HLMC\sqrt{d}T \geq T\left(1 - \alpha^2 HLMC\sqrt{d}\right) \geq 0
\end{align*}
It follows that $V_t \geq 0$ for all $t$. We also have that $V_0 = W_0$. Now for any $T > 0$:
\begin{align*}
    \mathbb{E}\left(W_0\right) & = \mathbb{E}\left(V_0\right) \geq \mathbb{E}\left(V_T\right) \\ & = \mathbb{E}\left(V_T|F_T\right)\mathbb{P}\left(F_T\right) + \mathbb{E}\left(V_T|\neg F_T\right)\mathbb{P}\left(\neg F_T\right) \\ & \geq \mathbb{E}\left(V_T|F_T\right)\mathbb{P}\left(F_T\right) \\ & = \mathbb{E}\left(W_T - \alpha^2 HLMC\sqrt{d}T + \alpha HL\sqrt{d}\sum_{k=1}^{T} \|x_{T-k+1} - x_{T-k}\|\sum_{m=k}^{\infty} \mathbbm{1}_{\{\tau_{T-k+m} \geq m\}}|F_T\right)\mathbb{P}\left(F_T\right) \\ & \geq \mathbb{E}\left(W_T - \alpha^2 HLMC\sqrt{d}T|F_T\right)\mathbb{P}\left(F_T\right) \\ & = \left(\mathbb{E}\left(W_T|F_T\right) - \alpha^2 HLMC\sqrt{d}T \right)\mathbb{P}\left(F_T\right)  \geq \left(1 - \alpha^2 HLMC\sqrt{d}\right)T\mathbb{P}\left(F_T\right). 
\end{align*}
We conclude that:
\begin{equation}
    \mathbb{P}\left(F_T\right) \leq \frac{\mathbb{E}\left(W_0\right)}{\left(1 - \alpha^2 HLMC\sqrt{d}\right)T}.
\end{equation}
\end{proof}

\subsection{Proof of Corollary~\ref{cor:exact_bound}}

\begin{repcorollary}{cor:exact_bound}
Assume that we run the lock-free SGD algorithm under the Asynchronous Shared Memory model for minimizing a convex function $f$ satisfying the listed assumptions. Set the learning rate to:
\begin{equation}
\alpha = \frac{c\epsilon\vartheta}{M^2 + 2\sqrt{\epsilon}LMC\sqrt{d}} = \frac{c\epsilon\vartheta}{M^2 + 4\sqrt{\epsilon}LM\sqrt{\tau_{\textrm{max}}n}\sqrt{d}},
\end{equation}
for some constant $\vartheta\in (0,1]$. Then for any $T > 0$ the probability that $x_i \not\in S$ for all $i \leq T$ is:
\begin{equation}
\mathbb{P}\left(F_T\right) \leq \frac{M^2 + 4\sqrt{\epsilon}LM\sqrt{\tau_{\textrm{max}}n}\sqrt{d}}{c^2 \epsilon\vartheta T}\textnormal{plog}\left(\frac{e\|x_0 - x^*\|^2}{\epsilon}\right).
\end{equation}
\end{repcorollary}
\begin{proof}
Substituting and using the result from~\cite{de2015taming} that $$\mathbb{E}\left(W_0\left(x_0\right)\right) \leq \frac{\epsilon}{2\alpha c\epsilon - \alpha^2 M^2}\textnormal{plog}\left(\frac{e \|x_0 - x^*\|^2}{\epsilon}\right)$$
\noindent we obtain that:
\begin{align*}
    \mathbb{P}\left(F_T\right) & \leq \frac{\mathbb{E}\left(W_0\right)}{\left(1 - \alpha^2 HLMC\sqrt{d}\right)T} \\ & \leq \frac{\epsilon}{2\alpha c\epsilon - \alpha^2 M^2}\textnormal{plog}\left(\frac{e \|x_0 - x^*\|^2}{\epsilon}\right)\left(\left(1 - \alpha^2 \frac{2\sqrt{\epsilon}}{2\alpha c\epsilon - \alpha^2 M^2} LMC\sqrt{d}\right)T\right)^{-1} \\ & \leq \frac{\epsilon}{\left(2\alpha c\epsilon - \alpha^2 \left(M^2 + 2\sqrt{\epsilon}LMC\sqrt{d}\right)  \right)T}\textnormal{plog}\left(\frac{e\|x_0 - x^*\|^2}{\epsilon}\right)
\end{align*}
Substituting the suggested value for the learning rate:
\begin{align*}
    \mathbb{P}\left(F_T\right) & \leq \frac{\epsilon}{T}\left(2c\epsilon\frac{c\epsilon\vartheta}{M^2 + 2\sqrt{\epsilon}LMC\sqrt{d}} - \left(M^2 + 2\sqrt{\epsilon}LMC\sqrt{d}\right) \left(\frac{c\epsilon\vartheta}{M^2 + 2\sqrt{\epsilon}LMC\sqrt{d}}\right)^2 \right)^{-1}\\ & \quad \textnormal{plog} \left(\frac{e\|x_0 - x^*\|^2}{\epsilon}\right) \\ & = \frac{\epsilon}{\left(\frac{2c^2 \epsilon^2\vartheta}{M^2 + 2\sqrt{\epsilon}LMC\sqrt{d}} - \frac{c^2 \epsilon^2\vartheta^2}{M^2 + 2\sqrt{\epsilon}LMC\sqrt{d}}\right)T} \textnormal{plog}\left(\frac{e\|x_0 - x^*\|^2}{\epsilon}\right)\\ &  \leq \frac{M^2 + 2\sqrt{\epsilon}LMC\sqrt{d}}{c^2 \epsilon\vartheta T}\textnormal{plog}\left(\frac{e\|x_0 - x^*\|^2}{\epsilon}\right) \\ & = \frac{M^2 + 4\sqrt{\epsilon}LM\sqrt{\tau_{\textrm{max}}n}\sqrt{d}}{c^2 \epsilon\vartheta T}\textnormal{plog}\left(\frac{e\|x_0 - x^*\|^2}{\epsilon}\right).
\end{align*}
\end{proof}

\subsection{Analysis of Algorithm~\ref{algo:sgd_iterate}} 

\begin{repcorollary}{cor:full}
    The \lit{FullSGD} procedure given in Algorithm~\ref{algo:sgd_iterate} guarantees that $\E [ \| r -  x^\star \| ] \leq \epsilon$ after executing for $O( T \log \left( \alpha {  2M n / \sqrt \epsilon } \right) )$ iterations. 
\end{repcorollary}
\begin{proof}[Proof Sketch]
The main idea behind the analysis is as follows. By Theorem \ref{thm:bound_on_fail_prob}, we know that, with high probability, there exists a time $t$ in each epoch where the aggregated gradients $x_t$ enter the success region, i.e. $\| x_t - x^* \|^2 \leq \epsilon$, for a fixed parameter $\epsilon$. 
The algorithm will guarantee that the model will not leave the success region before the end of the last epoch. This is ensured via our choice of the learning rate. 

Let us now focus on the last epoch. Fix an $\epsilon > 0$; we wish to prove that,  
$\| x_T - x^\star \|^2 \leq \epsilon$ at the end of this epoch, in expectation.
We fix the success condition of \lit{EpochSGD} such that that there exists an iteration $t$ in the epoch such that $\| x_t - x^\star\| \leq \sqrt \epsilon / 2$. 
We note that the adversary may attempt to schedule ``stale'' updates, generated earlier in the execution, to cause the algorithm to leave the success region. However, we notice that there can be at most $n - 1$ gradients generated \emph{before} time $t$, which have not been applied yet. 
Denote these gradients by $(G( v_{\theta_i} ))_{i = 1, \ldots, n - 1}$. 

Finally, we claim that, in expectation, the distance between the final model $x_T$ and the optimum is upper bounded by 
\begin{align*} \| x_T - x^* \| \leq \| x_t + \alpha \sum_{i = 1}^{n - 1} G( v_{\theta_i} ) - x^\star \| & \leq  \| x_T - x^* \| + \alpha n M \leq \sqrt \epsilon / 2 + \sqrt \epsilon / 2 = \sqrt \epsilon .
\end{align*}
\end{proof}

\end{document}